%% file: arxiv.tex
\documentclass[aps, prl, reprint, letterpaper, twocolumn, notitlepage, superscriptaddress, bibnotes]{revtex4-2}

\usepackage[charter, cal=cmcal, sfscaled=false]{mathdesign}
\usepackage[hidelinks, colorlinks=true, citecolor=blue, urlcolor=., linkcolor=.]{hyperref}
\usepackage{graphicx, braket, mathtools, appendix, amsmath, listings, xcolor, printlen, tikz, amsthm, ragged2e, comment, upgreek}
\usepackage[shortlabels]{enumitem} 
\usepackage[parfill]{parskip} 
\usepackage[mode=buildnew]{standalone} 

\usetikzlibrary{decorations.pathreplacing} 

\newtheorem{lemma}{Lemma} 
 
\newtheorem{theorem}{Theorem}
\newtheorem{example}{Example}

\newcommand\separator{\par\vskip-\lastskip\noindent$\vcenter{\hfil\hrule height .25mm}$\par}



\DeclareMathOperator{\tr}{tr} 

\justifying
\begin{document}
		
\title{The equivalence of quantum deletion and insertion errors on permutation-invariant codes}

\author{Lewis Bulled}
\email{lbulled1@sheffield.ac.uk}

\author{Yingkai Ouyang}
\email{y.ouyang@sheffield.ac.uk}

\affiliation{School of Mathematical and Physical Sciences, University of Sheffield, Sheffield, S3 7RH, United Kingdom}

\date{\today}

\begin{abstract}
    \noindent Quantum synchronisation errors are a class of errors that change the number of qubits in a quantum system. The classical error correction of synchronisation errors has been well-studied, including an insertion-deletion equivalence more than a half-century ago, but little progress has been made towards the quantum counterpart since the birth of quantum error correction. We address the longstanding problem of a quantum insertion-deletion equivalence on permutation-invariant codes, detailing the conditions under which such codes are $t$-insertion error-correctable. We extend these conditions to quantum insertion-deletion errors, formulating a more restrictive set of conditions under which permutation-invariant codes are $(t,s)$-insdel error-correctable. For completeness, we extend our approach to qudit errors on qudit permutation-invariant codes. Our work resolves many of the outstanding questions regarding the quantum error correction of synchronisation errors.
\end{abstract}

\maketitle

\input{main}

\onecolumngrid
\separator
\vspace{-1cm}
\section*{Appendix}
\label{Appendix}
\input{supp}
\separator
\vspace{0.5cm}
\twocolumngrid

\bibliography{ref}

\end{document}

%% file: main.tex
\section{Introduction}
\label{Sec:Introduction}

\noindent 
It is well-understood that quantum error correction (QEC) is fundamental to the success of future quantum technologies. It is thus no surprise that QEC literature is expansive, and a variety of quantum errors have been studied over the last three decades. A class of errors that has been historically under-researched, though, is synchronisation errors. Such errors change the number of qubits in a quantum system, including erasure, deletion and insertion errors. Erasure errors delete qubits in known locations, whereas deletion (insertion) errors delete (insert) qubits randomly according to some underlying probability distribution.

\noindent Synchronisation errors are prevalent in quantum communication, specifically optical systems. When states are transmitted across a lossy quantum channel, environmental noise may lead to qubit loss in the form of erasure or deletion errors~\cite{Ouyang_2021}. On the other hand, insertion errors may occur as a result of environmental noise or hardware imperfections, and it is theoretically possible for a combination of both error types to occur, that is, insertion-deletion (``insdel'') errors. 

\noindent Sixty years ago, Levenshtein demonstrated the equivalence between classical deletions and insertions~\cite{Levenshtein_1966}. In particular, Levenshtein showed that a $t$-deletion error-correctable code is a $t$-insertion error-correctable code, and vice versa. Research on classical insdel codes is reasonably mature~\cite{Schulman_Zuckerman_1997, Haeupler_2017, Duc_Liu_2021}, largely due to its applications in DNA storage and racetrack memories~\cite{Chee_2017, Buschmann_2013}, and progress is still being made at present~\cite{Segrest_2025}. Although preliminary attempts at quantum insdel channels have been made~\cite{Leahy_2019}, research on quantum insdel errors is generally quite limited.

\noindent Research on quantum deletion codes has gained significant momentum over the past several years. Nakayama and Hagiwara presented the first quantum deletion code -- an eight-qubit code capable of correcting a single deletion error~\cite{Nakayama_Hagiwara_Dec2019} -- and subsequently reduced the code length to four qubits~\cite{Hagiwara_Nakayama_Jan2020}. The authors argued that the optimal code length is four qubits, that is, there exists no quantum code on two or three qubits capable of correcting a single deletion error, and later provided a general construction of such codes~\cite{Hagiwara_Nakayama_Apr2020}. Shibayama and Hagiwara introduced a method for constructing permutation-invariant (PI) codes capable of correcting multiple deletions~\cite{Shibayama_Hagiwara_2021}. Ouyang noted an equivalence between erasure and deletion errors on PI codes~\cite{Ouyang_2021}, and proposed that such codes must have distance at least $t+1$ to correct $t$ deletions. Aydin \emph{et al.} proposed an infinite family of PI codes~\cite{Aydin_2024}, so-called combinatorial, AAB or $\mathcal{Q}_{g,m,\delta,\epsilon}$-codes, which can correct $2t$ deletion errors (and thus all $t$-qubit Pauli errors.) The author later extended this work to $SU(q)$ codes~\cite{Aydin_2025}, including spin, Fock state and qudit PI codes. Ouyang and Brennen introduced a novel QEC framework based on projective measurements of angular momentum~\cite{Ouyang_Brennen_2026}, providing a syndrome extraction protocol that corrects all Pauli and deletion errors, and Bulled and Ouyang later extended this to insertion errors~\cite{Bulled_2025}. 

\noindent Still, there remains a notable gap in the literature on quantum insertion codes. Hagiwara discovered the first quantum insertion code -- showing that the four-qubit deletion code also functions as an insertion code~\cite{Hagiwara_2021} -- and Shibayama and Hagiwara later provided a family of codes capable of correcting both single insertions and deletions~\cite{Shibayama_Hagiwara_2022}. Motivated by Levenshtein’s classical equivalence~\cite{Levenshtein_1966}, much of the recent work has focused on establishing a quantum analogue. Shibayama and Ouyang proved the equivalence of deletions and separable insertions~\cite{Shibayama_Ouyang_2021}, and later Shibayama reformulated the Knill–Laflamme QEC criterion~\cite{Knill_Laflamme_1996} for arbitrary single insertions~\cite{Shibayama_2025} (though multiple insertions were not considered.) Recently, Nakamura and Nozaki~\cite{Nakamura_2026} studied the relationship between insertion and deletion errors, showing that the equivalence does not hold for general QEC codes.

\noindent In this paper, we first prove the longstanding quantum insertion-deletion equivalence on PI codes. By utilising combinatorial techniques, we demonstrate that our $t$-insertion conditions are equivalent to the $t$-deletion conditions of Aydin \emph{et al.}~\cite[Thm. 4.1]{Aydin_2024}, which suffices to prove the equivalence. We then extend these ideas to quantum insdel errors, formulating the conditions under which PI codes are both $t$-insertion \emph{and} $s$-deletion error-correctable. Finally, we extend our both our $t$-insertion and $(t,s)$-insdel conditions to qudit errors on qudit PI codes. 

\section{Results}

\subsection{Permutation-invariant quantum codes}
\label{Sec:PI codes}

\noindent A class of quantum codes that is particularly well-suited for the QEC of synchronisation errors is PI codes. Such codes were introduced by Ruskai in 2000~\cite{Ruskai_2000}, and the literature has since matured greatly, cf.~\cite{Pollatsek_Ruskai_2004, Ouyang_2014, Ouyang_Fitzsimons_2016, Ouyang_2017, Hagiwara_Nakayama_Jan2020, Aydin_2024, Aydin_2025, Ouyang_Jing_Brennen_2025}. PI codes are defined within the symmetric space, i.e. the subspace of a given Hilbert space such that states remain invariant under any permutation of the underlying qubits. The $N$-qubit symmetric space is spanned by Dicke states
\begin{equation}
    \label{Dicke state}
    \ket{D^N_k} = \frac{1}{\sqrt{\binom{N}{k}}} \sum_{\substack{\mathbf{x} \in \{0,1\}^N, \\ |\mathbf{x}| = k}} \ket{\mathbf{x}},
\end{equation}
where $\binom{N}{k}$ denotes the binomial coefficient. We consider a PI code with logical codewords
\begin{equation}
    \label{PI logical codewords}
    \ket{0_L} = \sum_{k=0}^N \alpha_k \ket{D^N_k}, \qquad \ket{1_L} = \sum_{k=0}^N \beta_k \ket{D^N_k}.
\end{equation}
We define \emph{logical coefficients} $\vec{\alpha} \coloneqq (\alpha_0, \ldots, \alpha_N)$ and $\vec{\beta} \coloneqq (\beta_0, \ldots, \beta_N)$, where $\alpha_k, \beta_k \in \mathbb{C}$ are such that $\vec{\alpha}, \vec{\beta}$ are orthonormal:
\begin{equation}
    \label{Orthonormality}
    \sum_{k=0}^N \alpha^*_k \beta_k = 0, \qquad \sum_{k=0}^N |\alpha_k|^2 = \sum_{k=0}^N |\beta_k|^2 = 1.
\end{equation}
We then encode a single qubit $c_0 \ket{0} + c_1 \ket{1}$ into an $N$-qubit logical state $\ket{\psi_N} = c_0 \ket{0_L} + c_1 \ket{1_L}$, where $c_0, c_1 \in \mathbb{C}$ are such that $|c_0|^2 + |c_1|^2 = 1$. It is often convenient to write the logical state as a linear combination of Dicke states:
\begin{equation}
    \label{Logical state}
    \ket{\psi_N} = \sum_{k=0}^N \gamma_k \ket{D^N_k},
\end{equation}
where $\gamma_k \coloneqq c_0 \alpha_k + c_1 \beta_k$.

\subsection{Quantum insertion errors}
\label{Sec:Insertion error QEC}

\noindent Insertion errors occur when unwanted qubits materialise on the logical state \eqref{Logical state}. We consider the $t$-qubit insertion error 
\begin{equation}
    \label{t-qubit insertion}
    \ket{\phi_{\vec{v}}} = \sum_{x_1, \ldots, x_t \in \{0,1\}} v_{x_1 \cdots x_t} \ket{x_1 \cdots x_t},
\end{equation}
where $\sum_{x_1, \ldots, x_t} |v_{x_1 \cdots x_t}|^2 = 1$ and $\vec{v} \coloneqq (v_{x_1 \cdots x_t})_{x_1, \ldots, x_t}$ are \emph{insertion coefficients}. The \emph{insertion structure} is defined by $\vec{a} \coloneqq (a_0, \ldots, a_N)$, an $(N+1)$-tuple corresponding to the insertion of $a_i$ qubits in position $i$, where $0 \leq a_i \leq t$ and $0 \leq i \leq N$. Clearly $\sum_i a_i = t$, so $\vec{a}$ is an integer partition of $t$ into $N+1$ partitions. Note that there exist $\binom{N+t}{t} \sim \mathcal{O}(N^t)$ unique insertion structures, which is exponential in $t$. Given a PI state $\uprho_N$, probability distribution $p$ and corresponding measure $\mathrm{d}\mu$, we construct the $t$-insertion channel
\begin{equation}
    \label{t-insertion channel}
    \mathcal{I}_{t,p} (\uprho_N) = \sum_{\vec{a}} \int_{\vec{v}} p(\vec{a}, \vec{v}) \, \pi_{\vec{a}} \left( \ket{\phi_{\vec{v}}} \bra{\phi_{\vec{v}}} \otimes \uprho_N \right) \pi_{\vec{a}}^{\dagger} \, \mathrm{d}\mu.
\end{equation}
Here, $\pi_{\vec{a}}$ is an $(N+t)$-qubit matrix representation of a permutation operator that permutes qubits according to the insertion structure $\vec{a}$. Since a mixed state is nothing more a convex combination of pure states of the form $\sum_{\vec{v}} z_{\vec{v}} \ket{\phi_{\vec{v}}} \bra{\phi_{\vec{v}}}$, where $z_{\vec{v}} \geq 0$ and $\sum_{\vec{v}} z_{\vec{v}} = 1$, by linearity the $t$-insertion channel \eqref{t-insertion channel} extends to the insertion of $t$-qubit mixed states. Furthermore, \eqref{t-insertion channel} admits the Kraus decomposition
\begin{equation}  
    \label{Insertion Kraus decomp}
    \mathcal{I}_{t,p} (\uprho_N) = \sum_{\vec{a}} \int_{\vec{v}} K_{\vec{a}, \vec{v}} \uprho_N K^{\dagger}_{\vec{a}, \vec{v}} \, \mathrm{d}\mu,
\end{equation}
where $K_{\vec{a}, \vec{v}} \coloneqq \sqrt{p(\vec{a}, \vec{v})} \, \pi_{\vec{a}} \, (\ket{\phi_{\vec{v}}} \otimes \mathbb{1}_N)$ is a Kraus operator. Note that $0 < p(\vec{a}, \vec{v}) < 1$ for all $\vec{a}, \vec{v}$, since $p(\vec{a}, \vec{v}) = 0$ corresponds to no insertion errors, and $p(\vec{a}, \vec{v}) = 1$ corresponds to a fixed state inserted according to a fixed insertion structure (in which case, QEC is trivial.) For clarity, we provide an example of a single insertion error in Ex. \ref{Example} below. \newline

\begin{example}
    \label{Example}
    Consider the four-qubit PI code~\cite{Hagiwara_Nakayama_Jan2020} with logical state 
    \begin{equation}
        \label{Four-qubit code}
        \ket{\psi_4} = c_0 \left( \frac{\ket{D^4_0} + \ket{D^4_4}}{\sqrt{2}} \right) + c_1 \ket{D^4_2}.
    \end{equation}
    Suppose that an arbitrary single insertion error of the form $\ket{\phi_{\vec{v}}} = v_0 \ket{0} + v_1 \ket{1}$ occurs on \eqref{Four-qubit code}, where $\vec{v} = (v_0, v_1)$. There are five possible insertion structures, $\vec{a} \in \{\epsilon(1,0,0,0,0)\}$, where $\epsilon(\cdot)$ denotes all permutations of $\cdot$. Now consider the insertion position $\vec{a} = (0,0,1,0,0)$, that is, a single insertion error in the central position. In this case, we obtain 
    \begin{align}
        \label{Post-insertion state for a=(0,0,1,0,0)}
        &\frac{c_0}{\sqrt{2}}
        \Bigl( v_0 \ket{00000} + v_1 \ket{00100} + v_0 \ket{11011} + v_1\ket{11111} \Bigr) 
        \notag \\
        &+ \frac{c_1}{\sqrt{6}} \Bigl( v_0 \ket{00011} + v_1\ket{00111}
        + v_0\ket{01001} + v_1\ket{01101}
        \notag \\
        &\qquad + v_0\ket{01010} + v_1\ket{01110} + v_0\ket{10001} + v_1\ket{10101}
        \notag \\
        &\qquad + v_0\ket{10010} + v_1\ket{10110} + v_0\ket{11000} + v_1\ket{11100} \Bigr),
    \end{align}
    where we expanded the Dicke states in \eqref{Four-qubit code} in terms of the computational basis states using \eqref{Dicke state}. An analogous expression to \eqref{Post-insertion state for a=(0,0,1,0,0)} can be obtained for each of the four remaining insertion structures, and the post-insertion state is then given as a convex combination of these five states, with coefficients $p(\vec{a}, \vec{v})$ drawn from some probability distribution.
\end{example}

\noindent For a given $t$-insertion channel \eqref{t-insertion channel}, Thm. \ref{t-insertion conditions} states the conditions on $\vec{\alpha}, \vec{\beta}$ for which the PI code given by \eqref{PI logical codewords} is $t$-insertion error-correctable. \\

\begin{theorem}[$t$-insertion conditions]
    \label{t-insertion conditions}
    Let $\mathcal{C}$ be a PI code with logical coefficients $\vec{\alpha}, \vec{\beta}$. Then $\mathcal{C}$ is $t$-insertion error-correctable if and only if for all $0 \leq a', b' \leq t-j$ and $0 \leq j \leq t$, $\vec{\alpha}, \vec{\beta}$ satisfy
    \begin{align}
        \label{Condition C1}
        &\sum_{k=0}^{N-t+j} \frac{\binom{N-t+j}{k}}{\sqrt{\binom{N}{k+a'}\binom{N}{k+b'}}} \, \alpha^*_{k+a'} \beta_{k+b'} = 0, \\
        \label{Condition C2}
        &\sum_{k=0}^{N-t+j} \frac{\binom{N-t+j}{k}}{\sqrt{\binom{N}{k+a'}\binom{N}{k+b'}}} \left( \alpha^*_{k+a'}\alpha_{k+b'} - \beta^*_{k+a'}\beta_{k+b'} \right) = 0.
    \end{align}
\end{theorem}

\noindent Remarkably, \eqref{Condition C1}, \eqref{Condition C2} are equivalent to the $(t-j)$-deletion conditions of Aydin \emph{et al.} \cite[Thm. 4.1]{Aydin_2024} for all $0 \leq j \leq t$. Thm. \ref{t-insertion conditions} thus demonstrates the equivalence of quantum deletion and insertion errors on PI codes; that is, if a PI code is $t$-deletion error-correctable, then it is also $t$-insertion error-correctable (and vice versa.) Consequently, all deletion PI codes of distance at least $t+1$, such as $gnu$ codes~\cite{Ouyang_2014}, $\mathcal{Q}_{g,m\delta,\epsilon}$-codes~\cite{Aydin_2024} and $(b,g,m)$-codes~\cite{Ouyang_Jing_Brennen_2025}, are automatically insertion codes. Note that this does not hold for non-PI codes~\cite{Nakamura_2026}, since insertion error-correctability does not imply deletion error-correctability, in general.

\subsection{Quantum insdel errors}
\label{Sec:Insdel error QEC}

\noindent We now consider quantum insdel errors, where both deletion and insertion errors occur on the logical state \eqref{Logical state}. We first define a \emph{deletion set} $d_N \coloneqq \{ i \subset [N]: |i|=s \}$, where $[N] = \{1,\dots,N\}$, of which there exist $\binom{N}{s}$ such sets. Then given some probability distribution $q$, we introduce the $s$-deletion channel~\cite{Shibayama_Hagiwara_2021, Aydin_2024}
\begin{equation}
    \label{s-deletion channel}
    \mathcal{D}_{s,q} (\uprho_N) = \sum_{d \in d_N} q(d) \tr_d (\uprho_N),
\end{equation}
where $\tr_d$ denotes the partial trace over all $d \in d_N$. As for insertions, note that $0 < q(d) < 1$ for all $d$; $q(d) = 0$ corresponds to no deletions and $q(d) = 1$ corresponds to deletions in known locations, i.e. erasure errors. Furthermore, for all PI states $\uprho_N$, \eqref{s-deletion channel} reduces to the uniform deletion channel $\mathcal{D}_{s,\mathcal{U}}$, where $q(d) = \binom{N}{s}^{-1}$ for all $d \in d_N$ and $\mathcal{U}$ denotes the discrete uniform distribution. For PI input states, $\mathcal{D}_{s,\mathcal{U}}$ coincides with the partial trace over $s$ qubits. \\

\noindent By composing the channels \eqref{t-insertion channel} and \eqref{s-deletion channel}, we obtain an insdel channel $\mathcal{I}_{t,p} \circ \mathcal{D}_{s,\mathcal{U}}$ that represents $s$ deletions on a PI state followed by $t$ insertions. Physically, this describes only the deletion of qubits on the logical state (since insertions occur second), and we call such errors \emph{semi-insdel errors}. In Lem. \ref{Deletion-insertion conditions}, we show that semi-insdel errors follow similar error-correction conditions to those of $t$-insertion errors. \\

\begin{lemma}[Semi-insdel conditions]
    \label{Deletion-insertion conditions}
    Let $\mathcal{C}$ be a PI code with logical coefficients $\vec{\alpha}, \vec{\beta}$. Then $\mathcal{C}$ is semi-insdel error-correctable if and only if for all $0 \leq a', b' \leq t+s-j$ and $0 \leq j \leq t$, $\vec{\alpha}, \vec{\beta}$ satisfy
    \begin{align}
        \label{Condition C3}
        &\sum_{k=0}^{N-t-s+j} \frac{\binom{N-t-s+j}{k}}{\sqrt{\binom{N}{k+a'}\binom{N}{k+b'}}} \, \alpha^*_{k+a'} \beta_{k+b'} = 0, \\
        \label{Condition C4}
        &\sum_{k=0}^{N-t-s+j} \frac{\binom{N-t-s+j}{k}}{\sqrt{\binom{N}{k+a'}\binom{N}{k+b'}}} \left( \alpha^*_{k+a'}\alpha_{k+b'} - \beta^*_{k+a'}\beta_{k+b'} \right) = 0.
    \end{align}
\end{lemma}

\noindent We omit the proof of Lem. \ref{Deletion-insertion conditions} for brevity; see the Supplementary Information for full details. Now consider the insdel channel $\mathcal{D}_{s,q} \circ \mathcal{I}_{t,p}$, which represents $t$ insertions followed by $s$ deletions. In this case, both inserted qubits and qubits on the logical state may be deleted (since deletions occur second), and we call such errors \emph{full-insdel errors}. In Lem. \ref{Channel swap}, we prove a crucial result regarding the commutativity of semi-insdel and full-insdel channels. \\

\begin{lemma}
    \label{Channel swap}
    Let $\mathcal{I}_{t,p}$ be the $t$-insertion channel in \eqref{t-insertion channel} and $\mathcal{D}_{s,q}$ be the $s$-deletion channel in \eqref{s-deletion channel}. Then for all PI states $\uprho_N$, there exist probability distributions $r$ and $y$ such that
    \begin{equation}
        \label{Equivalence of channels}
        \mathcal{D}_{s,q} \circ \mathcal{I}_{t,p} (\uprho_N) = \sum_{\ell=0}^{\min(s,t)} r(\ell) \,\, \mathcal{I}_{t-\ell, y} \circ \mathcal{D}_{s-\ell, \mathcal{U}} \, (\uprho_N).
    \end{equation}
\end{lemma}

\noindent Lem. \ref{Channel swap} demonstrates that the insertion of $t$ qubits followed by the deletion of $s$ qubits is equivalent to a convex combination of the deletion of $s-\ell$ qubits followed by the insertion of $t-\ell$ qubits, where $0 \leq \ell \leq \text{min}(s,t)$. This highlights that \textbf{the only non-trivial interaction between insertions and deletions is the number of deleted inserted qubits.} In Thm. \ref{(s_1+s_2,t)-insdel conditions}, we utilise Lem. \ref{Deletion-insertion conditions} and \ref{Channel swap} to show that full-insdel errors follow similar error-correction conditions to both semi-insdel and $t$-insertion errors. \\

\begin{theorem}[Full-insdel conditions]
    \label{(s_1+s_2,t)-insdel conditions}
    Let $\mathcal{C}$ be a PI code with logical coefficients $\vec{\alpha}, \vec{\beta}$. Then $\mathcal{C}$ is full-insdel error-correctable if and only if for all $0 \leq a', b' \leq t+s-2\ell-j$, $0 \leq j \leq t-\ell$ and $0 \leq \ell \leq \mathrm{min}(s,t)$, $\vec{\alpha}, \vec{\beta}$ satisfy
    \begin{align}
        \label{Condition C5}
        &\sum_{k=0}^{N-t-s+2\ell+j} \frac{\binom{N-t-s+2\ell+j}{k}}{\sqrt{\binom{N}{k+a'}\binom{N}{k+b'}}} \, \alpha^*_{k+a'}\beta_{k+b'} = 0, \\
        \label{Condition C6}
        &\sum_{k=0}^{N-t-s+2\ell+j} \frac{\binom{N-t-s+2\ell+j}{k}}{\sqrt{\binom{N}{k+a'}\binom{N}{k+b'}}} \left( \alpha^*_{k+a'}\alpha_{k+b'} - \beta_{k+a'}^*\beta_{k+b'} \right) = 0.
    \end{align}
\end{theorem}

\subsection{Qudit codes and errors}

\noindent For completeness, we extend our qubit insertion, semi-insdel and full-insdel conditions to qudit PI codes with local dimension $q \geq 3$ and logical dimension $K \geq 3$. The $N$-qudit symmetric space is spanned by generalised Dicke states
\begin{equation}
    \label{Qudit Dicke state}
    \ket{D^N_{\underline{k}}} = \frac{1}{\sqrt{\binom{N}{{\underline{k}}}}} \sum_{\substack{\mathbf{x} \in \{0, \ldots, q-1\}^N, \\ C(\mathbf{x}) = \underline{k}}} \ket{\mathbf{x}},
\end{equation}
where the \emph{type} or \emph{composition} of $\mathbf{x} \in \{0, \ldots, q-1\}^N$ is given by the tuple $C(\mathbf{x}) \coloneqq (k_0, \ldots, k_{q-1})$~\cite{Aydin_2025}. Here, $\sum_j k_j = N$ and
\begin{equation}
    \binom{N}{{\underline{k}}} \coloneqq \frac{N!}{k_0! \cdots k_{q-1}!}
\end{equation}
denotes the multinomial coefficient. The set of all integer partitions of $N$ into at most $q$ parts is given by the discrete simplex
\begin{equation}
    \mathcal{S}_{q,N} \coloneqq \big\{ \underline{k} \in \mathbb{Z}^q : \sum_{j=0}^{q-1} k_j = N \big\}.
\end{equation}
For $0 \leq j \leq K-1$, we consider a qudit PI code with logical codewords
\begin{equation}
    \label{Qudit PI logical codewords}
    \ket{j_L} = \sum_{{\underline{k}} \in \mathcal{S}_{q,N}} \alpha^{(j)}_{\underline{k}} \ket{D^N_{\underline{k}}}.
\end{equation}
We define logical coefficients $\vec{\alpha}_j \coloneqq \big( \alpha^{(j)}_{\underline{k}} \big)_{{\underline{k}} \in \mathcal{S}_{q,N}}$ such that $\vec{\alpha}_j$ are orthonormal for all $0 \leq j \leq K-1$:
\begin{equation}
    \label{Qudit orthonormality}
    \sum_{\underline{k} \in \mathcal{S}_{q,N}} \big( \alpha^{(i)}_{\underline{k}} \big)^* \alpha^{(j)}_{\underline{k}} = \delta_{ij},
\end{equation}
where $\delta_{ij}$ is the Kronecker delta function. We encode a single qudit of dimension $q \geq 2$, $\sum_j c_j \ket{j}$, into an $N$-qudit logical state of dimension $K \geq 2$, $\ket{\psi_N} = \sum_j c_j \ket{j_L}$, where $\sum_j |c_j|^2 = 1$. As in \eqref{Logical state}, we may write the logical state as a linear combination of Dicke states:
\begin{equation}
    \label{Qudit logical state}
    \ket{\psi_N} = \sum_{\underline{k} \in \mathcal{S}_{q,N}} \gamma_{\underline{k}} \ket{D^N_{\underline{k}}},
\end{equation}
where $\gamma_{\underline{k}} \coloneqq \sum_j c_j \alpha^{(j)}_{\underline{k}}$.

\noindent The analysis of qubit insertion errors generalises directly to the qudit case, where a $t$-qudit insertion error is given by
\begin{equation}
    \label{t-qudit insertion}
    \ket{\phi_{\vec{v}}} = \sum_{x_1, \ldots, x_t \in \{0, \ldots, q-1\}^N} v_{x_1 \cdots x_t} \ket{x_1 \cdots x_t}.
\end{equation} 
In Thm. \ref{Qudit t-insertion conditions}, we extend Thm. \ref{t-insertion conditions} to form a set of $K(K-1)$ conditions under which the qudit PI code \eqref{Qudit PI logical codewords} is qudit $t$-insertion error-correctable. As in the qubit case, \eqref{Condition C7}, \eqref{Condition C8} are equivalent to the qudit $t$-deletion conditions of Aydin \emph{et al.}~\cite[Fig. 4]{Aydin_2025}, which proves the quantum insertion-deletion equivalence for qudit PI codes. \\

\begin{theorem}[Qudit $t$-insertion conditions]
    \label{Qudit t-insertion conditions}
    Let $\mathcal{C}_K$ be a qudit PI code with logical coefficients $\vec{\alpha}_j$, where $0 \leq j \leq K-1$. Then $\mathcal{C}_K$ is qudit $t$-insertion error-correctable if and only if for all $i \neq j$, $0 \leq n_{ij} \leq t$ and $\underline{e}, \underline{e}' \in \mathcal{S}_{q, t-n_{ij}}$, $\vec{\alpha}_j$ satisfy
    \begin{align}
        \label{Condition C7}
        &\sum_{\underline{k} \in \mathcal{S}_{q,N}} \frac{\binom{N-t+n_{ij}}{\underline{k}-\underline{e}}}{\sqrt{\binom{N}{\underline{k}} \binom{N}{\underline{k}+\underline{e}'-\underline{e}}}} \, \big( \alpha^{(i)}_{\underline{k}} \big)^* \alpha^{(j)}_{\underline{k}+\underline{e}'-\underline{e}} = 0, \\
        \label{Condition C8}
        &\sum_{\underline{k} \in \mathcal{S}_{q,N}} \frac{\binom{N-t+n_{ij}}{\underline{k}-\underline{e}}}{\sqrt{\binom{N}{\underline{k}} \binom{N}{\underline{k}+\underline{e}'-\underline{e}}}} \, \left( \big( \alpha^{(i)}_{\underline{k}} \big)^* \alpha^{(i)}_{\underline{k}+\underline{e}'-\underline{e}} - \big( \alpha^{(j)}_{\underline{k}} \big)^* \alpha^{(j)}_{\underline{k}+\underline{e}'-\underline{e}} \right) = 0.
    \end{align}
\end{theorem}

\noindent The proof of Thm. \ref{Qudit t-insertion conditions} can be found in the Supplementary Information. In a similar fashion, qubit insdel errors extend naturally to the qudit case. In Lem. \ref{Qudit semi-insdel conditions}, we state the conditions under which \eqref{Qudit PI logical codewords} is qudit semi-insdel error-correctable; again, the proof can be found in the Supplementary Information. Similarly, in Thm. \ref{Qudit full-insdel conditions} we state the conditions under which \eqref{Qudit PI logical codewords} is qudit full-insdel error correctable; we omit the proof since it is identical to Thm. \ref{(s_1+s_2,t)-insdel conditions}. Finally, we note that Lem. \ref{Channel swap} extends to the qudit case. \\

\begin{lemma}[Qudit semi-insdel conditions]
    \label{Qudit semi-insdel conditions}
    Let $\mathcal{C}_K$ be a PI qudit code with logical coefficients $\vec{\alpha}_j$, where $0 \leq j \leq K-1$. Then $\mathcal{C}_K$ is qudit semi-insdel error-correctable if and only if for all $i \neq j$, $0 \leq n_{ij} \leq t$ and $\underline{e}, \underline{e}', \underline{f}, \underline{f}' \in \mathcal{S}_{q, t+s-n_{ij}}$, $\vec{\alpha}_j$ satisfy
    \begin{align}
        \label{Condition C9} 
        &\sum_{\underline{k} \in \mathcal{S}_{q,N}} \frac{\binom{N-t-s+n_{ij}}{\underline{k}-\underline{e}-\underline{e}'}}{\sqrt{\binom{N}{\underline{k}} \binom{N}{\underline{k}+\underline{f}+\underline{f}'-\underline{e}-\underline{e}'}}} \, \big( \alpha^{(i)}_{\underline{k}} \big)^* \alpha^{(j)}_{\underline{k}+\underline{f}+\underline{f}'-\underline{e}-\underline{e}'} = 0, \\
        \begin{split}
            \notag
            &\sum_{\underline{k} \in \mathcal{S}_{q,N}} \frac{\binom{N-t-s+n_{ij}}{\underline{k}-\underline{e}-\underline{e}'}}{\sqrt{\binom{N}{\underline{k}} \binom{N}{\underline{k}+\underline{f}+\underline{f}'-\underline{e}-\underline{e}'}}} \, \left( \big( \alpha^{(i)}_{\underline{k}} \big)^* \alpha^{(i)}_{\underline{k}+\underline{f}+\underline{f}'-\underline{e}-\underline{e}'} \right.
        \end{split} \\
        \begin{split}
            \label{Condition C10}
            \hspace{3.2cm} \left. - \big( \alpha^{(j)}_{\underline{k}} \big)^*  \alpha^{(j)}_{\underline{k}+\underline{f}+\underline{f}'-\underline{e}-\underline{e}'} \right) = 0.
        \end{split}
    \end{align}
\end{lemma}

\begin{theorem}[Qudit full-insdel conditions]
    \label{Qudit full-insdel conditions}
    Let $\mathcal{C}_K$ be a PI qudit code with logical coefficients $\vec{\alpha}_j$, where $0 \leq j \leq K-1$. Then $\mathcal{C}_K$ is qudit full-insdel error-correctable if and only if for all $i \neq j$, $0 \leq \ell \leq \mathrm{min}(s,t)$, $0 \leq n_{ij} \leq t-\ell$ and $\underline{e}, \underline{e}', \underline{f}, \underline{f}' \in \mathcal{S}_{q, t+s-2\ell-n_{ij}}$, $\vec{\alpha}_j$ satisfy
    \begin{align}
        \label{Condition C11}
        &\sum_{\underline{k} \in \mathcal{S}_{q,N}} \frac{\binom{N-t-s+2\ell+n_{ij}}{\underline{k}-\underline{e}-\underline{e}'}}{\sqrt{\binom{N}{\underline{k}} \binom{N}{\underline{k}+\underline{f}+\underline{f}'-\underline{e}-\underline{e}'}}} \, \big( \alpha^{(i)}_{\underline{k}} \big)^* \alpha^{(j)}_{\underline{k}+\underline{f}+\underline{f}'-\underline{e}-\underline{e}'} = 0, \\
        \begin{split}
            \notag 
            &\sum_{\underline{k} \in \mathcal{S}_{q,N}} \frac{\binom{N-t-s+2\ell+n_{ij}}{\underline{k}-\underline{e}-\underline{e}'}}{\sqrt{\binom{N}{\underline{k}} \binom{N}{\underline{k}+\underline{f}+\underline{f}'-\underline{e}-\underline{e}'}}} \, \left( \big( \alpha^{(i)}_{\underline{k}} \big)^* \alpha^{(i)}_{\underline{k}+\underline{f}+\underline{f}'-\underline{e}-\underline{e}'} \right.
        \end{split} \\
        \begin{split}
            \label{Condition C12} 
            \hspace{3.2cm} \left. - \big( \alpha^{(j)}_{\underline{k}} \big)^*  \alpha^{(j)}_{\underline{k}+\underline{f}+\underline{f}'-\underline{e}-\underline{e}'} \right) = 0.
        \end{split}
    \end{align}
\end{theorem}

\section{Methods}
\label{Sec:Methods}

\subsection{Proof of Thm. \ref{t-insertion conditions}}

\begin{proof}
    We show that \eqref{Condition C1}, \eqref{Condition C2} are equivalent to the Knill-Laflamme error-correction criterion \cite{Knill_Laflamme_1996} for the $t$-insertion channel \eqref{t-insertion channel}. Let $0 \leq j \leq t$ be the size of the overlap between the insertion structures of two different insertion branches. Then by the Vandermonde decomposition, we can write
    \begin{equation}
        \ket{D^N_k} = \sum_{a'=0}^{t-j} \sqrt{\frac{\binom{t-j}{a'} \binom{N-t+j}{k-a'}}{\binom{N}{k}}} \ket{D^{t-j}_{a'}} \ket{D^{N-t+j}_{k-a'}}.
    \end{equation}
    Thus given probability distributions $p, p'$, there exists an $(N+t)$-qubit matrix representation $\sigma$ of a permutation operator such that
    \begin{multline}
        \sigma \, K_{\vec{a}, \vec{v}} \ket{0_L} = \sqrt{p(\vec{a}, \vec{v})} \sum_{\substack{x_1, \ldots, x_t, \\ k, a'}} \sqrt{\frac{\binom{t-j}{a'} \binom{N-t+j}{k-a'}}{\binom{N}{k}}} \, \alpha_k v_{x_1 \cdots x_t} \\ \ket{x_1 \cdots x_{t-j}} \ket{x_{t-j+1} \cdots x_t} \ket{D^{t-j}_{a'}} \ket{D^{N-t+j}_{k-a'}},
    \end{multline}
    \vspace{-0.6cm}
    \begin{multline}
        \sigma \, K_{\vec{b}, \vec{u}} \ket{1_L} = \sqrt{p'(\vec{b}, \vec{u})} \sum_{\substack{x_1', \ldots, x_t', \\ k', b'}} \sqrt{\frac{\binom{t-j}{b'} \binom{N-t+j}{k'-b'}}{\binom{N}{k'}}} \, \beta_{k'} u_{x_1' \cdots x_t'} \\ \ket{D^{t-j}_{b'}} \ket{x_1' \cdots x_j'} \ket{x_{j+1}' \cdots x_t'} \ket{D^{N-t+j}_{k'-b'}}.
    \end{multline}
    See Fig. \ref{Fig:Permuted post-insertion state} for a visual depiction and further detail. Note that $\sigma$ permutes the $t$ inserted qubits, but we can relabel these without any loss of generality. The Knill-Laflamme orthogonality condition $\bra{0_L} K^{\dagger}_{\vec{a}, \vec{v}} K_{\vec{b}, \vec{u}} \ket{1_L} = 0$ \cite{Knill_Laflamme_1996} yields
    \begin{equation}
        \label{KL orthogonality}
        \sum_{\substack{x_1, x_1', \ldots, x_t, x_t', \\ k, k', a', b'}} \Delta \sqrt{\frac{\binom{t-j}{a'} \binom{t-j}{b'} \binom{N-t+j}{k-a'} \binom{N-t+j}{k'-b'}}{\binom{N}{k} \binom{N}{k'}}} \, \alpha^*_k \beta_{k'} v^*_{x_1 \cdots x_t} u_{x_1' \cdots x_t'} = 0,
        \end{equation}
    where $\Delta$ defines the following product of inner products:
    \begin{multline}
        \label{Inner products}
        \Delta \coloneqq \braket{x_1 \cdots x_{t-j} | D^{t-j}_{b'}} \braket{D^{t-j}_{a'} | x_{j+1}' \cdots x_t'} \\ \braket{D^{N-t+j}_{k-a'} | D^{N-t+j}_{k'-b'}} \braket{x_{t-j+1} \cdots x_t | x_1' \cdots x_j'}.
    \end{multline}
    One can verify that \eqref{Inner products} simplifies to
    \begin{equation}
        \label{Delta functions}
        \Delta = \frac{\delta_{k', k+b'-a'} \delta_{x_{t-j+1}, x_1'} \cdots \delta_{x_t, x'_j}}{\sqrt{\binom{t-j}{a'} \binom{t-j}{b'}}},
    \end{equation}
    where $b' = \sum_{i=1}^{t-j} x_i$ and $a' = \sum_{i=j+1}^t x_i'$ necessarily. By applying the Kronecker deltas in \eqref{Delta functions} and rearranging, \eqref{KL orthogonality} simplifies to
    \begin{equation}
        \label{Intermediate eqn}
        \sum_{\substack{x_1, \ldots, x_{t-j}, \\ x_1', \ldots, x_t'}} v^*_{x_1 \cdots x_{t-j} x_1', \ldots, x_j'} u_{x_1' \cdots x_t'} \sum_{k=0}^N \frac{\binom{N-t+j}{k-a'}}{\sqrt{\binom{N}{k} \binom{N}{k+b'-a'}}} \, \alpha^*_k \beta_{k+b'-a'} = 0,
    \end{equation}
    where we view the outer sum over $x_1, \ldots, x_{t-j}, x_1', \ldots, x_t'$ as a coefficient of the inner sum over $k$. Since the inserted states $\ket{\phi_{\vec{u}}}, \ket{\phi_{\vec{v}}}$ are arbitrary, \eqref{Intermediate eqn} must hold for all choices of $\vec{u}, \vec{v}$. Thus the inner sum must vanish:
    \begin{equation}
        \sum_{k=0}^N \frac{\binom{N-t+j}{k-a'}}{\sqrt{\binom{N}{k} \binom{N}{k+b'-a'}}} \, \alpha^*_k \beta_{k+b'-a'} = 0.
    \end{equation}
    Finally, reindexing $k \rightarrow k+a'$ yields \eqref{Condition C1}:
    \begin{equation}
        \sum_{k=0}^{N-t+j} \frac{\binom{N-t+j}{k}}{\sqrt{\binom{N}{k+a'} \binom{N}{k+b'}}} \, \alpha^*_{k+a'} \beta_{k+b'} = 0,
    \end{equation}
    where the summation range changes since $a' \in \{0, \ldots, t-j\}$. In a similar fashion, from the Knill-Laflamme non-deformation condition $\bra{0_L} K^{\dagger}_{\vec{a}, \vec{v}} K_{\vec{b}, \vec{u}} \ket{0_L} = \bra{1_L} K^{\dagger}_{\vec{a}, \vec{v}} K_{\vec{b}, \vec{u}} \ket{1_L}$ \cite{Knill_Laflamme_1996} we obtain \eqref{Condition C2}:
    \begin{equation}
        \sum_{k=0}^{N-t+j} \frac{\binom{N-t+j}{k}}{\sqrt{\binom{N}{k+a'} \binom{N}{k+b'}}} \left( \alpha^*_{k+a'} \alpha_{k+b'} - \beta^*_{k+a'} \beta_{k+b'} \right) = 0.
    \end{equation}
\end{proof}

\begin{figure}
    \centering
    \includegraphics[scale=0.74]{fig1.tex}
    \caption[Permuted post-insertion state]{\textbf{Action of permuted Kraus operators on logical codewords.} The top row depicts $\sigma \, K_{\vec{a}, \vec{v}} \ket{0_L}$, and likewise the bottom row $\sigma \, K_{\vec{b}, \vec{u}} \ket{1_L}$. The permutation $\sigma \coloneqq \sigma_3 \circ \sigma_2 \circ \sigma_1$ is such that the qubits overlap by $j$, where $0 \leq j \leq t$, with the qubit count underbraced. Here, $\sigma_1$ acts non-trivially on the full $(N+t)$-qubit register and arranges the $t$ inserted qubits on the top row leftwards; $\sigma_2$ then acts non-trivially on the first $t$ qubits and arranges $j$ inserted qubits on the bottom row rightwards; finally, $\sigma_3$ acts non-trivially on the last $N$ qubits and arranges the remaining $t-j$ inserted qubits on the bottom row leftwards.}
    \label{Fig:Permuted post-insertion state}
\end{figure}

\subsection{Proof of Lem. \ref{Channel swap}}

\begin{proof}
    We show that the composition of insertion and deletion channels decomposes according to the number of inserted qubits that are deleted, which allows for a commutation between the resulting channels. We first apply the $s$-deletion channel \eqref{s-deletion channel} to \eqref{t-insertion channel}, which yields
    \begin{multline}
        \label{Insdel channel}
        \mathcal{D}_{s,q} \circ \mathcal{I}_{t,p} (\uprho_N) = \sum_{d \in d_{N+t}} \sum_{\vec{a}} \int_{\vec{v}} p(\vec{a}, \vec{v}) \, q(d) \\ 
        \tr_d \left( \pi_{\vec{a}} \left( \ket{\phi_{\vec{v}}} \bra{\phi_{\vec{v}}} \otimes \uprho_N \right) \pi_{\vec{a}}^{\dagger} \right) \, \mathrm{d}\mu.
    \end{multline}
    For a quantum state $\sigma$, we introduce the channel $\mathcal{Q}_{\vec{a},d} (\sigma) \coloneqq \tr_d \big( \pi_{\vec{a}} \sigma \pi^\dagger_{\vec{a}} \big)$. We also define the joint distribution $w(\vec{a},d,\vec{v}) \coloneqq p(\vec{a}, \vec{v}) \, q(d)$ so that \eqref{Insdel channel} can be written as
    \begin{multline}
        \label{Insdel channel simplified}
        \mathcal{D}_{s,q} \circ \mathcal{I}_{t,p} (\uprho_N) = \sum_{d \in d_{N+t}} \sum_{\vec{a}} \int_{\vec{v}} w(\vec{a},d,\vec{v}) \\ \mathcal{Q}_{\vec{a},d} \left( \ket{\phi_{\vec{v}}} \bra{\phi_{\vec{v}}} \otimes \uprho_N \right) \, \mathrm{d}\mu',
    \end{multline}
    for a joint measure $\mathrm{d}\mu'$. Now let $\ell$ count the number of inserted qubits that are deleted, where $0 \leq \ell \leq \min(s,t)$. Thus $s-\ell$ counts the number of deleted qubits on the logical state whenever $s>\ell$. For a given insertion structure $\vec{a}$, each deletion set $d_{N+t}$ induces a \emph{reduced insertion structure} $\vec{a}' \coloneqq (a_0', \ldots, a_{N-s+\ell}')$ such that $\sum_i a_i' = t-\ell$, where $0 \leq a_i'  \leq t-\ell$ and $0 \leq i \leq N-s+\ell$. Thus for all $\vec{a}$ and $d \in d_{N+t}$, there exists a surjection $\varphi$ such that $\varphi(\vec{a}, d) = (\ell, \vec{a}')$. Such a surjection induces the pushforward distribution 
    \begin{equation}
        \label{Pushforward}
        \tilde w(\ell,\vec{a}',\vec{v}) \coloneqq \sum_{(\vec{a},d): \,\varphi(\vec{a},d) = (\ell,\vec{a}')} w(\vec{a},d,\vec{v}).
    \end{equation}
    Note that for every pair $(\vec a, d)$, we can ``commute'' the partial trace $\tr_d$ and permutation $\pi_{\vec a}$ in the sense that 
    \begin{equation}
        \label{Commute partial trace past permutation}
        \mathcal{Q}_{\vec a, d} (\ket{\phi_{\vec{v}}} \bra{\phi_{\vec{v}}} \otimes \uprho_N) = \pi_{\vec a'} \left( \tr_{d'}(\ket{\phi_{\vec{v}}} \bra{\phi_{\vec{v}}}) \otimes \tr_{[s-\ell]}(\uprho_N) \right)
        \pi_{\vec a'}^\dagger  
    \end{equation}
    for some permutation $\pi_{\vec a'}$ and deletion set $d'$ such that $|d'|=\ell$. Noting that $\mathcal{D}_{s-\ell, \mathcal U} (\uprho_N) = \tr_{[s-\ell]} (\uprho_N)$ for all PI states $\uprho_N$, if we define the $(t-\ell)$-qubit mixed state $\sigma_{\vec v,d} \coloneqq \tr_{d'}(\ket{\phi_{\vec{v}}} \bra{\phi_{\vec{v}}})$, then \eqref{Commute partial trace past permutation} becomes
    \begin{align}
        \label{Commute tr_d and pi_a}
        \mathcal{Q}_{\vec a, d} (\ket{\phi_{\vec{v}}} \bra{\phi_{\vec{v}}} \otimes \uprho_N) = \pi_{\vec a'} \left( \sigma_{\vec v, d'} \otimes \mathcal D_{s-\ell, \mathcal U}(\uprho_N) \right) \pi_{\vec a'}^\dagger.
    \end{align}
    For all $0 \leq \ell \leq \min(s,t)$, let us define the set of all reduced insertion structures $\mathcal{A}_\ell \coloneqq \{ \vec{a}': (\ell, \vec{a}') \in \operatorname{Im}(\varphi) \}$. Using \eqref{Pushforward} and \eqref{Commute tr_d and pi_a}, we can rewrite \eqref{Insdel channel simplified} as
    \begin{multline}
        \label{Commuted channel}
        \mathcal{D}_{s,q} \circ \mathcal{I}_{t,p} (\uprho_N) = \sum_{\ell=0}^{\min(s,t)} \sum_{\vec a' \in \mathcal{A}_\ell} 
        \int_{\vec v} \tilde w(\ell, \vec a', \vec v) \\ 
        \pi_{\vec a'} \left( \sigma_{\vec v ,d'} \otimes \mathcal D_{s-\ell, \mathcal U}(\uprho_N) \right) \pi_{\vec a'}^\dagger \, \mathrm{d}\mu'',
    \end{multline}
    for a pushforward measure $\mathrm{d}\mu''$. Now let us define the marginal distribution
    \begin{equation}
        r(\ell) \coloneqq \sum_{\vec a' \in \mathcal{A}_\ell} \int_{\vec v} \tilde w(\ell, \vec a', \vec v) \, \mathrm{d}\mu''
    \end{equation}
    and consequently the conditional distribution $y(\vec a', \vec v \,|\, \ell) \coloneqq \tilde w(\ell, \vec a', \vec v) / r(\ell)$. Substituting these back into \eqref{Commuted channel} yields
    \begin{multline}
        \label{Final form 1}
        \mathcal{D}_{s,q} \circ \mathcal{I}_{t,p} (\uprho_N) = \sum_{\ell=0}^{\min(s,t)} r(\ell) \sum_{\vec a' \in \mathcal{A}_\ell} 
        \int_{\vec v} y(\vec a', \vec v \,|\, \ell) \\ 
        \pi_{\vec a'} \left( \sigma_{\vec v, d'} \otimes \mathcal D_{s-\ell, \mathcal U}(\uprho_N) \right) \pi_{\vec a'}^\dagger \, \mathrm{d}\mu''',
    \end{multline}
    for a conditional measure $\mathrm{d}\mu'''$. The inner sum on the RHS of \eqref{Final form 1} can be identified with the channel $\mathcal{I}_{t-\ell, y} \circ \mathcal{D}_{s-\ell, \mathcal{U}} \, (\uprho_N)$, which completes the proof.
\end{proof}

\subsection{Proof of Thm. \ref{(s_1+s_2,t)-insdel conditions}}

\begin{proof}
    Recall from the proof of Lem. \ref{Deletion-insertion conditions} that the channel $\mathcal{I}_{t,p} \circ \mathcal{D}_{s,\mathcal{U}}$ has Kraus operators of the form $K_{\vec{a}, \vec{v}} \tilde K_a$. Thus each semi-insdel channel $\mathcal{I}_{t-\ell,y} \circ \mathcal{D}_{s-\ell,\mathcal{U}}$ has Kraus operators of the form $K_{\vec{a}', \vec{v}'} \tilde K_{a'}$ for some $\vec{a}', \vec{v}', a'$, where we replace $s \rightarrow s-\ell, t \rightarrow t-\ell$. Let us define $\mathcal{E}_{\vec{a}', \vec{v}', a'} \coloneqq K_{\vec{a}', \vec{v}'} \tilde K_{a'}$ so that the semi-insdel channel admits the Kraus decomposition 
    \begin{equation}
        \mathcal{I}_{t-\ell,y} \circ \mathcal{D}_{s-\ell,\mathcal{U}} (\uprho_N) = \sum_{\vec{a}', a'} \int_{\vec{v}'} \mathcal{E}_{\vec{a}', \vec{v}', a'} \uprho_N \mathcal{E}^\dag_{\vec{a}', \vec{v}', a'} \, \mathrm{d}\mu,
    \end{equation}
    for an appropriate measure $\mathrm{d}\mu$ (note that this is not necessarily the same measure as in the proof of Lem. \ref{Channel swap}.) Then by Lem. \ref{Channel swap}, the full-insdel channel admits the Kraus decomposition
    \begin{equation}
        \label{Full-insdel Kraus decomp}
         \mathcal{D}_{s,q} \circ \mathcal{I}_{t,p} (\uprho_N) = \sum_{\ell, \vec{a}', a'} \int_{\vec{v}'} \mathcal{E}^\ell_{\vec{a}', \vec{v}', a'} \uprho_N \big( \mathcal{E}^\ell_{\vec{a}', \vec{v}', a'} \big)^\dag \, \mathrm{d}\mu',
    \end{equation}
    where $\mathcal{E}^\ell_{\vec{a}', \vec{v}', a'} \coloneqq \sqrt{r(\ell)} \,\, \mathcal{E}_{\vec{a}', \vec{v}', a'}$ and $\mathrm{d}\mu'$ is a corresponding measure. For the channel \eqref{Full-insdel Kraus decomp} to be correctable, its Kraus operators must satisfy the Knill-Laflamme error correction criterion~\cite{Knill_Laflamme_1996}
    \begin{align}
        \label{KL conditions}
        &\bra{0_L} \big( \mathcal{E}^{\ell_1}_{\vec{a}', \vec{v}', a'} \big)^\dag \mathcal{E}^{\ell_2}_{\vec{b}', \vec{u}', b'} \ket{1_L} = 0, \\
        &\notag \bra{0_L} \big( \mathcal{E}^{\ell_1}_{\vec{a}', \vec{v}', a'} \big)^\dag \mathcal{E}^{\ell_2}_{\vec{b}', \vec{u}', b'} \ket{0_L} - \bra{1_L} \big( \mathcal{E}^{\ell_1}_{\vec{a}', \vec{v}', a'} \big)^\dag \mathcal{E}^{\ell_2}_{\vec{b}', \vec{u}', b'} \ket{1_L} = 0
    \end{align}
    for all $0 \leq \ell_1, \ell_2 \leq \min(s,t)$. Notice that we may write
    \begin{equation} 
        \big( \mathcal{E}^{\ell_1}_{\vec{a}', \vec{v}', a'} \big)^\dag \mathcal{E}^{\ell_2}_{\vec{b}', \vec{u}', b'} = \sqrt{r(\ell_1) \, r(\ell_2)} \,\, \tilde K^\dag_{a'} K^\dag_{\vec{a}', \vec{v}'} K_{\vec{b}', \vec{u}'} \tilde K_{b'},
    \end{equation}
    so \eqref{KL conditions} is equivalent to 
    \begin{align}
        \label{Reduced KL conditions}
        &\bra{0_L} \tilde K^\dag_{a'} K^\dag_{\vec{a}', \vec{v}'} K_{\vec{b}', \vec{u}'} \tilde K_{b'} \ket{1_L} = 0, \\
        &\notag \bra{0_L} \tilde K^\dag_{a'} K^\dag_{\vec{a}', \vec{v}'} K_{\vec{b}', \vec{u}'} \tilde K_{b'} \ket{0_L} - \bra{1_L} \tilde K^\dag_{a'} K^\dag_{\vec{a}', \vec{v}'} K_{\vec{b}', \vec{u}'} \tilde K_{b'} \ket{1_L} = 0.
    \end{align}
    We proved in Lem. \ref{Deletion-insertion conditions} that \eqref{Reduced KL conditions} are satisfied with $(\vec{a}', \vec{v}', a') \,\rightarrow\, (\vec{a}, \vec{v}, a)$ if and only if $\vec{\alpha}, \vec{\beta}$ satisfy \eqref{Condition C3}, \eqref{Condition C4}. This corresponds to the substitution $N-t-s+j \,\rightarrow\, N-t-s+\ell_1+\ell_2+j$, which yields
    \begin{align}
        \label{l_1,l_2 conditions}
        &\sum_{k=0}^{N-t-s+\ell_1+\ell_2+j} \frac{\binom{N-t-s+\ell_1+\ell_2+j}{k}}{\sqrt{\binom{N}{k+a'}\binom{N}{k+b'}}} \, \alpha^*_{k+a'}\beta_{k+b'} = 0, \\
        &\notag \sum_{k=0}^{N-t-s+\ell_1+\ell_2+j} \frac{\binom{N-t-s+\ell_1+\ell_2+j}{k}}{\sqrt{\binom{N}{k+a'}\binom{N}{k+b'}}} \left( \alpha^*_{k+a'}\alpha_{k+b'} - \beta_{k+a'}^*\beta_{k+b'} \right) = 0.
    \end{align}
    
    \noindent Clearly, we recover \eqref{Condition C5}, \eqref{Condition C6} from \eqref{l_1,l_2 conditions} whenever $\ell_1 = \ell_2$; physically, this corresponds to two arbitrary insertion branches experiencing the same number of deleted inserted qubits. We now argue the converse, i.e. that \eqref{l_1,l_2 conditions} contains \eqref{Condition C5}, \eqref{Condition C6}. If we define
    \begin{equation}
        \ell^* \coloneqq \min(\ell_1,\ell_2), \qquad j^* \coloneqq j+|\ell_1-\ell_2|,
    \end{equation}
    then
    \begin{equation}
        2\ell^*+j^* = \ell_1+\ell_2+j
    \end{equation}
    and thus
    \begin{equation}
        N-t-s+\ell_1+\ell_2+j = N-t-s+2\ell^*+j^*.
    \end{equation}
    Furthermore, since $0 \leq j \leq \min(t-\ell_1, t-\ell_2)$ and $\min(t-\ell_1, t-\ell_2) = t-\max(\ell_1,\ell_2)$, it follows that
    \begin{equation}
        0 \leq j^* \leq t-\ell^*,
    \end{equation}
    and so $\ell^*,j^*$ lie within the required ranges. Thus \eqref{l_1,l_2 conditions} is contained within \eqref{Condition C5}, \eqref{Condition C6} via some admissible $\ell^*,j^*$. Consequently, the channel $\mathcal{D}_{s,q} \circ \mathcal{I}_{t,p}$ is correctable with respect to $\mathcal{C}$ if and only if $\vec{\alpha}, \vec{\beta}$ satisfy \eqref{Condition C5}, \eqref{Condition C6}, as required.
\end{proof}

\section{Discussions}
\label{Sec:Discussions}

\noindent In this paper, we derived a set of $t$-insertion conditions and proved the quantum insertion-deletion equivalence for PI codes. We also analysed quantum insdel errors and derived a set of $(t,s)$-insdel conditions, before extending our qubit analyses to qudit codes and errors of local dimension $q \geq 2$ and logical dimension $K \geq 2$. Far less is understood about qudit PI codes than their qubit counterpart, so using our conditions \eqref{Condition C7}, \eqref{Condition C8}, \eqref{Condition C9}, \eqref{Condition C10}, \eqref{Condition C11}, \eqref{Condition C12} to search for new qudit codes may be a fruitful research avenue. 

\section{Data availability}
\label{Sec:Data availability}

\noindent No data were created or analysed in this study.

\section{Acknowledgements}
\label{Sec:Acknowledgements}

\noindent L.B. acknowledges support from EPSRC under grant number EP/W524360/1. Y.O. acknowledges support from EPSRC under grant number EP/W028115/1, and the EPSRC-funded QCI3 Hub (grant number EP/Z53318X/1.)

\section{Author contributions}
\label{Sec:Author contributions}

\noindent The project was conceived by all authors. Y.O. provided oversight and ideas. L.B. wrote the manuscript and created Fig. \ref{Fig:Permuted post-insertion state}. All authors discussed and reviewed the manuscript.

\section{Competing interests}

\noindent All authors declare no competing interests.

%% file: supp.tex
\section{Proof of Lem. \ref{Deletion-insertion conditions}}

\begin{proof}
    We show that \eqref{Condition C3}, \eqref{Condition C4} are equivalent to the Knill-Laflamme error-correction criterion~\cite{Knill_Laflamme_1996} for the semi-insdel channel $\mathcal{I}_{t,p} \circ \mathcal{D}_{s,\mathcal{U}}$. For all PI states $\uprho_N$, the $s$-qubit uniform deletion channel $\mathcal{D}_{s,\mathcal{U}}$ admits the Kraus decomposition 
    \begin{equation}
        \label{PI Kraus decomp}
        \mathcal{D}_{s,\mathcal{U}} (\uprho_N) = \sum_{a=0}^s \tilde K_a \uprho_N \tilde K^{\dagger}_a,
    \end{equation}
    where for all \emph{deletion weights} $0 \leq a \leq s$~\cite[Lem. 3.6]{Aydin_2024},
    \begin{equation}
        \label{Kraus operator action}
        \tilde K_a \ket{D^N_k} = \sqrt{\frac{\binom{s}{a}\binom{N-s}{k-a}}{\binom{N}{k}}} \ket{D^{N-s}_{k-a}}.
    \end{equation} 
    Note that \eqref{Kraus operator action} follows directly from~\cite[Lem. 4]{Ouyang_2021} and is given in~\cite[Lem. 3.6]{Aydin_2024}, though a factor of $\sqrt{\binom{s}{a}}$ is missing from the latter. By the Vandermonde decomposition, we can write
    \begin{equation}
        \ket{D^{N-s}_{k-a}} = \sum_{a''=0}^{t-j} \sqrt{\frac{\binom{t-j}{a''} \binom{N-s-t+j}{k-a-a''}}{\binom{N-s}{k-a}}} \ket{D^{t-j}_{a''}} \ket{D^{N-s-t+j}_{k-a-a''}}.
    \end{equation}
    Thus given probability distributions $p, p'$, there exists an $(N-s+t)$-qubit matrix representation $\epsilon$ of a permutation operator such that
    \begin{align}
        \epsilon \, K_{\vec{a}, \vec{v}} \tilde K_a \ket{0_L} &= \sqrt{\binom{s}{a} \, p(\vec{a}, \vec{v}}) \sum_{\substack{x_1, \ldots, x_t, \\ k, a''}} \sqrt{\frac{\binom{t-j}{a''} \binom{N-s-t+j}{k-a-a''}}{\binom{N}{k}}} \, \alpha_k v_{x_1 \cdots x_t} \ket{x_1 \cdots x_{t-j}} \ket{x_{t-j+1} \cdots x_t} \ket{D^{t-j}_{a''}} \ket{D^{N-s-t+j}_{k-a-a''}}, \\
        \epsilon \, K_{\vec{b}, \vec{u}} \tilde K_b \ket{1_L} &= \sqrt{\binom{s}{b} \, p'(\vec{b}, \vec{u}}) \sum_{\substack{x_1, \ldots, x_t, \\ k, b''}} \sqrt{\frac{\binom{t-j}{b''} \binom{N-s-t+j}{k-b-b''}}{\binom{N}{k}}} \, \beta_k u_{x_1 \cdots x_t} \ket{D^{t-j}_{b''}} \ket{x_1 \cdots x_j} \ket{x_{j+1} \cdots x_t} \ket{D^{N-s-t+j}_{k'-b-b''}}.
    \end{align}
    The Knill-Laflamme orthogonality condition $\bra{0_L} \tilde K_a^{\dagger} K_{\vec{a}, \vec{v}}^{\dagger} K_{\vec{b}, \vec{u}} \tilde K_b \ket{1_L} = 0$ \cite{Knill_Laflamme_1996} yields
    \begin{equation}
        \label{KL orthogonality insdel}
        \sum_{\substack{x_1, x_1', \ldots, x_t, x_t', \\ k, k', a'', b''}} \Delta \, \sqrt{\frac{\binom{t-j}{a''} \binom{t-j}{b''} \binom{N-s-t+j}{k-a-a''} \binom{N-s-t+j}{k'-b-b''}}{\binom{N}{k} \binom{N}{k'}}} \, \alpha^*_k \beta_{k'} v^*_{x_1 \cdots x_t} u_{x_1' \cdots x_t'} = 0,
    \end{equation}
    where $\Delta$ defines the following product of inner products:
    \begin{equation}
        \label{Inner products insdel}
        \Delta \coloneqq \braket{x_1 \cdots x_{t-j} | D^{t-j}_{b''}}  \braket{D^{t-j}_{a''} | x_{j+1}' \cdots x_t'} \braket{D^{N-s-t+j}_{k-a-a'} | D^{N-s-t+j}_{k'-b-b''}} \braket{x_{t-j+1} \cdots x_t | x_1' \cdots x_j'}.
    \end{equation}
    One can verify that \eqref{Inner products insdel} simplifies to
    \begin{equation}
        \label{Delta functions insdel}
        \Delta = \frac{\delta_{k', k+b+b''-a-a''} \delta_{x_{t-j+1}, x_1'} \cdots \delta_{x_t, x'_j}}{\sqrt{\binom{t-j}{a''} \binom{t-j}{b''}}},
    \end{equation}
    where $b'' = \sum_{i=1}^{t-j} x_i$ and $a'' = \sum_{i=j+1}^t x_i'$ necessarily. By applying the Kronecker deltas in \eqref{Delta functions insdel} and rearranging, \eqref{KL orthogonality insdel} becomes
    \begin{equation}
        \sum_{\substack{x_1, \ldots, x_{t-j}, \\ x_1', \ldots, x_t'}} v^*_{x_1 \cdots x_{t-j} x_1', \ldots, x_j'} u_{x_1' \cdots x_t'} \sum_{k=0}^N \frac{\binom{N-s-t+j}{k-a-a''}}{\sqrt{\binom{N}{k} \binom{N}{k+b+b''-a-a''}}} \, \alpha^*_k \beta_{k+b+b''-a-a''} = 0,
    \end{equation}
    where we view outer sum over $x_1, \ldots, x_{t-j}, x_1', \ldots, x_t'$ as a coefficient of the inner sum over $k$. Since the inserted states $\ket{\phi_{\vec{u}}}, \ket{\phi_{\vec{v}}}$ are arbitrary, \eqref{Intermediate eqn} must hold for all choices of $\vec{u}, \vec{v}$. Thus the inner sum must vanish:
    \begin{equation}
        \sum_{k=0}^N \frac{\binom{N-s-t+j}{k-a-a''}}{\sqrt{\binom{N}{k} \binom{N}{k+b+b''-a-a''}}} \, \alpha^*_k \beta_{k+b+b''-a-a''} = 0.
    \end{equation}
    By defining $a' \coloneqq a+a''$ and $b' \coloneqq b+b''$ and reindexing $k \rightarrow k+a'$, we obtain \eqref{Condition C3}:
    \begin{equation}
        \sum_{k=0}^{N-t-s+j} \frac{\binom{N-s-t+j}{k}}{\sqrt{\binom{N}{k+a'} \binom{N}{k+b'}}} \, \alpha^*_{k+a'} \beta_{k+b'} = 0,
    \end{equation}
    where the summation range changes since $a \in \{0, \ldots, s\}$ and $a'' \in \{0, \ldots, t-j\}$. In a similar fashion, the Knill-Laflamme non-deformation condition $\bra{0_L} \tilde K_a^{\dagger} K_{\vec{a}, \vec{v}}^{\dagger} K_{\vec{b}, \vec{u}} \tilde K_b \ket{0_L} = \bra{1_L} \tilde K_a^{\dagger} K_{\vec{a}, \vec{v}}^{\dagger} K_{\vec{b}, \vec{u}} \tilde K_b \ket{1_L}$ \cite{Knill_Laflamme_1996} yields \eqref{Condition C4}:
    \begin{equation}
        \sum_{k=0}^{N-t-s+j} \frac{\binom{N-s-t+j}{k}}{\sqrt{\binom{N}{k+a'} \binom{N}{k+b'}}} \left( \alpha^*_{k+a'} \alpha_{k+b'} - \beta^*_{k+a'} \beta_{k+b'} \right) = 0.
    \end{equation}
\end{proof}

\section{Proof of Thm. \ref{Qudit t-insertion conditions}}

\begin{proof}
    We show that \eqref{Condition C7}, \eqref{Condition C8} are equivalent to the Knill-Laflamme error-correction criterion \cite{Knill_Laflamme_1996} for the $t$-insertion channel \eqref{t-insertion channel}. Let $\underline{e} \in \mathcal{S}_{q,t-n_{ij}}$ denote the \emph{overlap structure} between two insertion branches with overlap of size $0 \leq n_{ij} \leq \min(s,t)$. By the generalised Vandermonde decomposition, we can write
    \begin{equation}
        \ket{D^N_{\underline{k}}} = \sum_{\underline{e} \in \mathcal{S}_{q,t-n_{ij}}} \sqrt{\frac{\binom{t-n_{ij}}{\underline{e}} \binom{N-t+n_{ij}}{\underline{k}-\underline{e}}}{\binom{N}{\underline{k}}}} \ket{D^{t-n_{ij}}_{\underline{e}}} \ket{D^{N-t+n_{ij}}_{\underline{k}-\underline{e}}}.
    \end{equation}
    Thus given probability distributions $p, p'$, there exists an $(N+t)$-qudit matrix representation $\sigma$ of a permutation operator such that
    \begin{equation}
        \sigma \, K_{\vec{a}, \vec{v}} \ket{i_L} = \sqrt{p(\vec{a}, \vec{v})} \sum_{\substack{x_1, \ldots, x_t, \\ \underline{k}, \underline{e}}} \sqrt{\frac{\binom{t-n_{ij}}{\underline{e}} \binom{N-t+n_{ij}}{\underline{k}-\underline{e}}}{\binom{N}{\underline{k}}}} \, \alpha^{(i)}_{\underline{k}} v_{x_1 \cdots x_t} \ket{x_1 \cdots x_{t-n_{ij}}} \ket{x_{t-n_{ij}+1} \cdots x_t} \ket{D^{t-n_{ij}}_{\underline{e}}} \ket{D^{N-t+n_{ij}}_{\underline{k}-\underline{e}}},
    \end{equation}
    \vspace{-0.6cm}
    \begin{equation}
        \sigma \, K_{\vec{b}, \vec{u}} \ket{j_L} = \sqrt{p'(\vec{b}, \vec{u})} \sum_{\substack{x_1', \ldots, x_t', \\ \underline{k}', \underline{e}'}} \sqrt{\frac{\binom{t-n_{ij}}{\underline{e}'} \binom{N-t+n_{ij}}{\underline{k}'-\underline{e}'}}{\binom{N}{\underline{k}'}}} \, \alpha^{(j)}_{\underline{k}'} u_{x_1' \cdots x_t'} \ket{D^{t-n_{ij}}_{\underline{e}'}} \ket{x_1' \cdots x_{n_{ij}}'} \ket{x_{n_{ij}+1}' \cdots x_t'} \ket{D^{N-t+n_{ij}}_{\underline{k}'-\underline{e}'}},
    \end{equation}
    where $i \neq j$. The Knill-Laflamme orthogonality condition $\bra{i_L} K^{\dagger}_{\vec{a}, \vec{v}} K_{\vec{b}, \vec{u}} \ket{j_L} = 0$ \cite{Knill_Laflamme_1996} yields
    \begin{equation}
        \label{Qudit KL orthogonality}
        \sum_{\substack{x_1, x_1', \ldots, x_t, x_t', \\ \underline{k}, \underline{k}', \underline{e}, \underline{e}'}} \Delta \sqrt{\frac{\binom{t-n_{ij}}{\underline{e}} \binom{t-n_{ij}}{\underline{e}'} \binom{N-t+n_{ij}}{\underline{k}-\underline{e}} \binom{N-t+n_{ij}}{\underline{k}'-\underline{e}'}}{\binom{N}{\underline{k}} \binom{N}{\underline{k}'}}} \, \big( \alpha^{(i)}_{\underline{k}} \big)^* \alpha^{(j)}_{\underline{k}'}  v^*_{x_1 \cdots x_t} u_{x_1' \cdots x_t'} = 0,
        \end{equation}
    where $\Delta$ defines the following product of inner products:
    \begin{equation}
        \label{Qudit inner products}
        \Delta \coloneqq \braket{x_1 \cdots x_{t-n_{ij}} | D^{t-n_{ij}}_{\underline{e}'}} \braket{D^{t-n_{ij}}_{\underline{e}} | x_{n_{ij}+1}' \cdots x_t'} \braket{D^{N-t+n_{ij}}_{\underline{k}-\underline{e}} | D^{N-t+n_{ij}}_{\underline{k}'-\underline{e}'}} \braket{x_{t-n_{ij}+1} \cdots x_t | x_1' \cdots x_{n_{ij}}'}.
    \end{equation}
    One can verify that \eqref{Qudit inner products} simplifies to
    \begin{equation}
        \label{Qudit delta functions}
        \Delta = \frac{\delta_{\underline{k}', \underline{k}+\underline{e}'-\underline{e}} \delta_{x_{t-n_{ij}+1}, x_1'} \cdots \delta_{x_t, x'_{n_{ij}}}}{\sqrt{\binom{t-n_{ij}}{\underline{e}} \binom{t-n_{ij}}{\underline{e}'}}}.
    \end{equation}
    By applying the Kronecker deltas in \eqref{Qudit delta functions} and rearranging, \eqref{Qudit KL orthogonality} simplifies to
    \begin{equation}
        \sum_{\substack{x_1, \ldots, x_{t-n_{ij}}, \\ x_1', \ldots, x_t'}} v^*_{x_1 \cdots x_{t-n_{ij}} x_1', \ldots, x_{n_{ij}}'} u_{x_1' \cdots x_t'} \sum_{\underline{k}} \frac{\binom{N-t+n_{ij}}{\underline{k}-\underline{e}}}{\sqrt{\binom{N}{\underline{k}} \binom{N}{\underline{k}+\underline{e}'-\underline{e}}}} \, \big( \alpha^{(i)}_{\underline{k}} \big)^* \alpha^{(j)}_{\underline{k}+\underline{e}'-\underline{e}} = 0,
    \end{equation}
    where we view the outer sum over $x_1, \ldots, x_{t-n_{ij}}, x_1', \ldots, x_t' \in \{0, \ldots, q-1\}$ as a coefficient of the inner sum over $\underline{k} \in \mathcal{S}_{q,N}$. Since the inserted states $\ket{\phi_{\vec{u}}}, \ket{\phi_{\vec{v}}}$ are arbitrary, \eqref{Intermediate eqn} must hold for all choices of $\vec{u}, \vec{v}$. Thus the inner sum must vanish:
    \begin{equation}
        \sum_{\underline{k}} \frac{\binom{N-t+n_{ij}}{\underline{k}-\underline{e}}}{\sqrt{\binom{N}{\underline{k}} \binom{N}{\underline{k}+\underline{e}'-\underline{e}}}} \, \big( \alpha^{(i)}_{\underline{k}} \big)^* \alpha^{(j)}_{\underline{k}+\underline{e}'-\underline{e}} = 0,
    \end{equation}
    which is exactly \eqref{Condition C7}. In a similar fashion, from the Knill-Laflamme non-deformation condition $\bra{i_L} K^{\dagger}_{\vec{a}, \vec{v}} K_{\vec{b}, \vec{u}} \ket{i_L} = \bra{j_L} K^{\dagger}_{\vec{a}, \vec{v}} K_{\vec{b}, \vec{u}} \ket{j_L}$ \cite{Knill_Laflamme_1996} we obtain \eqref{Condition C8}:
    \begin{equation}
        \sum_{\underline{k}} \frac{\binom{N-t+n_{ij}}{\underline{k}-\underline{e}}}{\sqrt{\binom{N}{\underline{k}} \binom{N}{\underline{k}+\underline{e}'-\underline{e}}}} \, \left( \big( \alpha^{(i)}_{\underline{k}} \big)^* \alpha^{(i)}_{\underline{k}+\underline{e}'-\underline{e}} - \big( \alpha^{(j)}_{\underline{k}} \big)^* \alpha^{(j)}_{\underline{k}+\underline{e}'-\underline{e}} \right) = 0.
    \end{equation}
\end{proof}

\section{Proof of Lem. \ref{Qudit semi-insdel conditions}}

\begin{proof}
    \noindent We show that \eqref{Condition C9}, \eqref{Condition C10} are equivalent to the Knill-Laflamme error-correction criterion \cite{Knill_Laflamme_1996} for the semi-insdel channel $\mathcal{I}_{t,p} \circ \mathcal{D}_{s,\mathcal{U}}$. For all PI states $\uprho_N$, the $s$-qudit uniform deletion channel $\mathcal{D}_{s,\mathcal{U}}$ admits the Kraus decomposition 
    \begin{equation}
        \label{Qudit PI Kraus decomp}
        \mathcal{D}_{s,\mathcal{U}} (\uprho_N) = \sum_{\underline{e} \in \mathcal{S}_{q,s}} \tilde K_{\underline{e}} \uprho_N \tilde K^{\dagger}_{\underline{e}},
    \end{equation}
    where for all $\underline{e} \in \mathcal{S}_{q,s}$~\cite[Lem. IV.4]{Aydin_2025},
    \begin{equation}
        \label{Qudit Kraus operator action}
        \tilde K_{\underline{e}} \ket{D^N_{\underline{k}}} = \sqrt{\frac{\binom{s}{\underline{e}} \binom{N-s}{\underline{k}-\underline{e}}}{\binom{N}{\underline{k}}}} \ket{D^{N-s}_{\underline{k}-\underline{e}}}.
    \end{equation} 
    Note that a factor of $\sqrt{\binom{s}{\underline{e}}}$ is missing from \cite{Aydin_2025}. By the generalised Vandermonde decomposition, we can write
    \begin{equation}
        \ket{D^{N-s}_{\underline{k}-\underline{e}}} = \sum_{\underline{e}' \in \mathcal{S}_{q, t+s-n_{ij}}} \sqrt{\frac{\binom{t-n_{ij}}{\underline{e}'} \binom{N-s-t+n_{ij}}{\underline{k}-\underline{e}-\underline{e}'}} {\binom{N-s}{\underline{k}-\underline{e}}}} \ket{D^{t-n_{ij}}_{\underline{e}'}} \ket{D^{N-s-t+n_{ij}}_{\underline{k}-\underline{e}-\underline{e}'}}.
    \end{equation}
    Thus given probability distributions $p, p'$, there exists an $(N-s+t)$-qudit matrix representation $\epsilon$ of a permutation operator such that
    \begin{equation}
        \epsilon \, K_{\vec{a}, \vec{v}} \tilde K_{\underline{e}} \ket{i_L} = \sqrt{\binom{s}{\underline{e}} \, p(\vec{a}, \vec{v}}) \sum_{\substack{x_1, \ldots, x_t, \\ \underline{k}, \underline{e}'}} \sqrt{\frac{\binom{t-n_{ij}}{\underline{e}'} \binom{N-s-t+n_{ij}}{\underline{k}-\underline{e}-\underline{e}'}} {\binom{N}{\underline{k}}}} \, \alpha^{(i)}_{\underline{k}} v_{x_1 \cdots x_t} \ket{x_1 \cdots x_{t-n_{ij}}} \ket{x_{t-n_{ij}+1} \cdots x_t} \ket{D^{t-n_{ij}}_{\underline{e}'}} \ket{D^{N-s-t+n_{ij}}_{\underline{k}-\underline{e}-\underline{e}'}},
    \end{equation}
    \vspace{-0.3cm}
    \begin{equation}
        \epsilon \, K_{\vec{b}, \vec{u}} \tilde K_{\underline{f}} \ket{j_L} = \sqrt{\binom{s}{\underline{f}} \, p'(\vec{b}, \vec{u}}) \sum_{\substack{x_1', \ldots, x_t', \\ \underline{k}', \underline{f}'}} \sqrt{\frac{\binom{t-n_{ij}}{\underline{f}'} \binom{N-s-t+n_{ij}}{\underline{k}'-\underline{f}-\underline{f}'}} {\binom{N}{\underline{k}'}}} \, \alpha^{(j)}_{\underline{k}'} u_{x_1' \cdots x_t'} \ket{D^{t-n_{ij}}_{\underline{f}'}} \ket{x_1' \cdots x_{n_{ij}}'} \ket{x_{n_{ij}'+1} \cdots x_t'} \ket{D^{N-s-t+n_{ij}}_{\underline{k}'-\underline{f}-\underline{f}'}},
    \end{equation}
    where $i \neq j$. The Knill-Laflamme orthogonality condition $\bra{i_L} \tilde K_{\underline{e}}^{\dagger} K_{\vec{a}, \vec{v}}^{\dagger} K_{\vec{b}, \vec{u}} \tilde K_{\underline{f}} \ket{j_L} = 0$ \cite{Knill_Laflamme_1996} yields
    \begin{equation}
        \label{Qudit KL orthogonality insdel}
        \sum_{\substack{x_1, x_1', \ldots, x_t, x_t', \\ \underline{k}, \underline{k}', \underline{e}', \underline{f}'}} \Delta \, \sqrt{\frac{\binom{t-n_{ij}}{\underline{e}'} \binom{t-n_{ij}}{\underline{f}'} \binom{N-s-t+n_{ij}}{\underline{k}-\underline{e}-\underline{e}'} \binom{N-s-t+n_{ij}}{\underline{k}'-\underline{f}-\underline{f}'}}{\binom{N}{\underline{k}} \binom{N}{\underline{k}'}}} \, \big( \alpha^{(i)}_{\underline{k}} \big)^* \alpha^{(j)}_{\underline{k}'} v^*_{x_1 \cdots x_t} u_{x_1' \cdots x_t'} = 0,
    \end{equation}
    where $\Delta$ defines the following product of inner products:
    \begin{equation}
        \label{Qudit inner products insdel}
        \Delta \coloneqq \braket{x_1 \cdots x_{t-n_{ij}} | D^{t-n_{ij}}_{\underline{f}'}} \braket{D^{t-n_{ij}}_{\underline{e}'} | x_{n_{ij}+1}' \cdots x_t'} \braket{D^{N-t+n_{ij}}_{\underline{k}-\underline{e}-\underline{e}'} | D^{N-t+n_{ij}}_{\underline{k}'-\underline{f}-\underline{f}'}} \braket{x_{t-n_{ij}+1} \cdots x_t | x_1' \cdots x_{n_{ij}}'}
    \end{equation}
    One can verify that \eqref{Qudit inner products insdel} simplifies to
    \begin{equation}
        \label{Qudit delta functions insdel}
        \Delta = \frac{\delta_{\underline{k}', \underline{k}+\underline{f}+\underline{f}'-\underline{e}-\underline{e}'} \delta_{x_{t-n_{ij}+1}, x_1'} \cdots \delta_{x_t, x'_{n_{ij}}}}{\sqrt{\binom{t-n_{ij}}{\underline{e}'} \binom{t-n_{ij}}{\underline{f}'}}}.
    \end{equation}
    By applying the Kronecker deltas in \eqref{Qudit delta functions insdel} and rearranging, \eqref{Qudit KL orthogonality insdel} becomes
    \begin{equation}
        \sum_{\substack{x_1, \ldots, x_{t-n_{ij}}, \\ x_1', \ldots, x_t'}} v^*_{x_1 \cdots x_{t-n_{ij}} x_1', \ldots, x_{n_{ij}}'} u_{x_1' \cdots x_t'} \sum_{\underline{k}} \frac{\binom{N-s-t+n_{ij}}{\underline{k}-\underline{e}-\underline{e}'}}{\sqrt{\binom{N}{\underline{k}} \binom{N}{\underline{k}+\underline{f}+\underline{f}'-\underline{e}-\underline{e}'}}} \, \big( \alpha^{(i)}_{\underline{k}} \big)^* \alpha^{(j)}_{\underline{k}+\underline{f}+\underline{f}'-\underline{e}-\underline{e}'} = 0,
    \end{equation}
    where we view the outer sum over $x_1, \ldots, x_{t-n_{ij}}, x_1', \ldots, x_t' \in \{0, \ldots, q-1\}$ as a coefficient of the inner sum over $\underline{k} \in \mathcal{S}_{q,N}$. Since the inserted states $\ket{\phi_{\vec{u}}}, \ket{\phi_{\vec{v}}}$ are arbitrary, \eqref{Intermediate eqn} must hold for all choices of $\vec{u}, \vec{v}$. Thus the inner sum must vanish:
    \begin{equation}
        \sum_{\underline{k}} \frac{\binom{N-s-t+n_{ij}}{\underline{k}-\underline{e}-\underline{e}'}}{\sqrt{\binom{N}{\underline{k}} \binom{N}{\underline{k}+\underline{f}+\underline{f}'-\underline{e}-\underline{e}'}}} \, \big( \alpha^{(i)}_{\underline{k}} \big)^* \alpha^{(j)}_{\underline{k}+\underline{f}+\underline{f}'-\underline{e}-\underline{e}'} = 0,
    \end{equation}
    which is exactly \eqref{Condition C9}. In a similar fashion, the Knill-Laflamme non-deformation condition $\bra{i_L} \tilde K_{\underline{e}}^{\dagger} K_{\vec{a}, \vec{v}}^{\dagger} K_{\vec{b}, \vec{u}} \tilde K_{\underline{f}} \ket{i_L} = \bra{j_L} \tilde K_{\underline{e}}^{\dagger} K_{\vec{a}, \vec{v}}^{\dagger} K_{\vec{b}, \vec{u}} \tilde K_{\underline{f}} \ket{j_L}$ \cite{Knill_Laflamme_1996} yields \eqref{Condition C10}:
    \begin{equation}
        \sum_{\underline{k}} \frac{\binom{N-s-t+n_{ij}}{\underline{k}-\underline{e}-\underline{e}'}}{\sqrt{\binom{N}{\underline{k}} \binom{N}{\underline{k}+\underline{f}+\underline{f}'-\underline{e}-\underline{e}'}}} \, \left( \big( \alpha^{(i)}_{\underline{k}} \big)^* \alpha^{(i)}_{\underline{k}+\underline{f}+\underline{f}'-\underline{e}-\underline{e}'} - \big( \alpha^{(j)}_{\underline{k}} \big)^* \alpha^{(j)}_{\underline{k}+\underline{f}+\underline{f}'-\underline{e}-\underline{e}'} \right) = 0.
    \end{equation}
\end{proof}